\newcommand{\nop}[1]{}
\documentclass[journal, twoside]{IEEEtran}
\pdfoutput=1

\ifCLASSINFOpdf
\else
\fi
\usepackage[pdftex]{graphicx}
\usepackage{booktabs} 
\usepackage{array} 
\usepackage{paralist} 
\usepackage{verbatim} 

\usepackage{float}
\usepackage{amsmath}
\usepackage{algorithm}
\usepackage{algpseudocode}
\usepackage{amsthm}
\usepackage{amssymb}
\usepackage{cite}
\usepackage{balance}
\usepackage{hyperref}

\renewcommand{\algorithmicrequire}{\textbf{Input:}} 
\renewcommand{\algorithmicensure}{\textbf{Output:}} 
\newtheorem{definition}{Definition}

\newtheorem{lemma}{Lemma}

\usepackage[belowskip=-2pt,aboveskip=2pt]{caption}
\ifCLASSOPTIONcompsoc
\usepackage[caption=false,font=normalsize,labelfont=sf,textfont=sf]{subfig}
\else
\usepackage[caption=false,font=footnotesize]{subfig}
\fi

\begin{document}
%
\title{From Landscape to Portrait: A New Approach for Outlier Detection in Load Curve Data}
%
%
%

\author{Guoming~Tang, Kui~Wu,~\IEEEmembership{Senior Member,~IEEE}, Jingsheng Lei, Zhongqin Bi, and Jiuyang Tang
}
%
%

\markboth{ieee transactions on smart grid,~Vol.~XX, No.~XX, XX~2013}{Tang \MakeLowercase{\textit{et al.}}: From Landscape to Portrait: A New Approach for Outlier Detection in Load Curve Data}
%



\maketitle

\begin{abstract}
In power systems, load curve data is one of the most important datasets that are collected and retained by utilities. The quality of load curve data, however, is hard to guarantee since the data is subject to communication losses, meter malfunctions, and many other impacts. In this paper, a new approach to analyzing load curve data is presented. The method adopts a new view, termed \textit{portrait}, on the load curve data by analyzing the periodic patterns in the data and re-organizing the data for ease of analysis. Furthermore, we introduce algorithms to build the virtual portrait load curve data, and demonstrate its application on load curve data cleansing. Compared to existing regression-based methods, our method is much faster and more accurate for both small-scale and large-scale real-world datasets.
\end{abstract}


%
\IEEEpeerreviewmaketitle
\vspace{-0.2in}
\section*{Nomenclature}

\textbf{Related to Portrait and Landscape Data}
\addcontentsline{toc}{section}{Nomenclature}
\begin{IEEEdescription}[\IEEEusemathlabelsep\IEEEsetlabelwidth{$out(\Theta,)$}]
\item[$t_i$] The $i$-th timestamp
\item[$y(t)$] The load curve data at time $t$
\item[$p_i$] The $i$-th basic portrait dataset (BPD)
\item[$P_i$] The $i$-th virtual portrait dataset (VPD)
\item[$l_i$] Landscape data in the $i$-th period
\item[$L_i$] The $i$-th virtual landscape data (VLD)
\item[${sp}_i$] Span of the $i$-th BPD or VPD
\item[$e_i$] Characteristic vector of the $i$-th BPD or VPD
\item[$\theta_i$] Median value of the $i$-th BPD or VPD
\item[$M_i$] Median absolute deviation (MAD) of the $i$-th BPD or VPD
\item[${s}_{ij}$] Similarity between the $i$-th and the $j$-th portrait datasets
\item[$d_{n}$] The mean distance of $n$ virtual portrait or landscape datasets\newline
\end{IEEEdescription}

\textbf{Related to Outlier Detection}
\addcontentsline{toc}{section}{Nomenclature}
\begin{IEEEdescription}[\IEEEusemathlabelsep\IEEEsetlabelwidth{$out(\Theta,)$}]
\item[$out(\alpha,\Theta)$] The outlier region of a certain distribution with parameter vector $\Theta$, with confidence coefficient or significance index $\alpha$
\item[$N(\mu,\sigma^{2})$] Normal distribution with mean $\mu$ and variance $\sigma^{2}$
\item[$G(\beta,\gamma)$] Gamma distribution with shape parameter $\beta$ and scale parameter $\gamma$
\item[$Q_{1},Q_{3}$] The lower quartile and upper quartile of a boxplot, respectively
\item[$IQR$] The interquartile range of a boxplot (\textit{i.e.}, $Q_{3}-Q_{1}$)
\item[$df$] The degree of freedom in B-spline smoothing\newline\newline
\end{IEEEdescription}

\section{Introduction}

\IEEEPARstart{I}{n} recent smart grid research~\cite{farhangi2010path, chen2009survey, chen2010automated}, load curve data, which refers to electric energy consumption data collected and retained by utilities, has become one of the most important datasets for a broad spectrum of applications. For electric utilities, the analysis of load curve data plays a significant role in day-to-day operations, system reliability, and energy planning. For the energy consumers, load curve data provides them with abundant information on their daily and seasonal energy cost, helping them make timely response to save expense. Overall, the importance of load curve data in the demand side management (DSM) of smart grid makes it the critical information in modern electric industry.

Due to the critical meaning of load curve data, its quality is of vital importance. Nevertheless, load curve data is subject to pollution caused by many factors, such as communication failures, meter malfunctions, unexpected interruption or shutdown of power stations, unscheduled maintenance, and temporary closure of production lines. In this paper, we call load curve data \textit{polluted} when it significantly deviates from its regular patterns or when some data items are missing. Due to its huge volume, it would be nearly impossible to manually identify the polluted load curve data. Clearly, an efficient, automatic method is needed to solve the {\em load curve data cleansing} problem, \textit{i.e.}, to detect and fix polluted load curve data.

\subsection{Motivation}

We have observed that all existing work arranges load curve data in chronological order, \textit{i.e.}, the load curve data is strictly treated as a time series. As shown in Fig.~\ref{fig_timelinePlot}, the hourly energy consumption of over one hundred residential houses was recorded for one year ($8760$ hours) and displayed in the 2D Coordinate System, with $x$-axis representing the time and $y$-axis the load values ($kWh$). In addition to aggregated data of multiple households, the low-level energy consumption of an individual household is illustrated in Fig.~\ref{fig_timelinePlot2}, which shows the similar periodic pattern\footnote{Note that some low-level household load curve data may not have periodic pattern at all. In this case, the method developed in this paper does not bring benefit.}. We call such type of arrangement of load data as {\em landscape data}. Landscape data is easy to understand, but it poses several barriers to efficient analysis. 

\nop{
\begin{figure}[!ht]
\begin{center}
\includegraphics[width=4in]{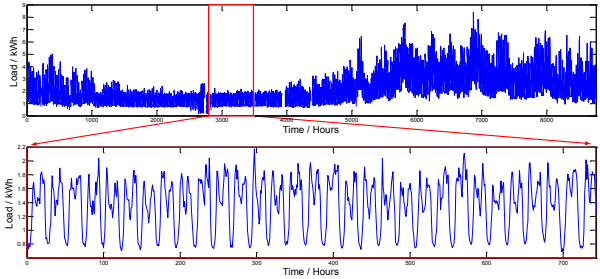}
\caption{Average energy consumption of $112$ residential houses in U.S. for one year from $01/04/2006$ to $31/03/2007$ (above) and data for one month from $01/08/2006$ to $31/08/2006$ (below), provided by Pacific Northwest National Laboratory~\cite{hammerstrom2007pacific}.}\label{fig_timelinePlot}
\end{center}
\vspace{-0.1in}
\end{figure} 
}

\begin{figure*}
\begin{center}
\begin{minipage}{3.2in}
\includegraphics[width=3.2in]{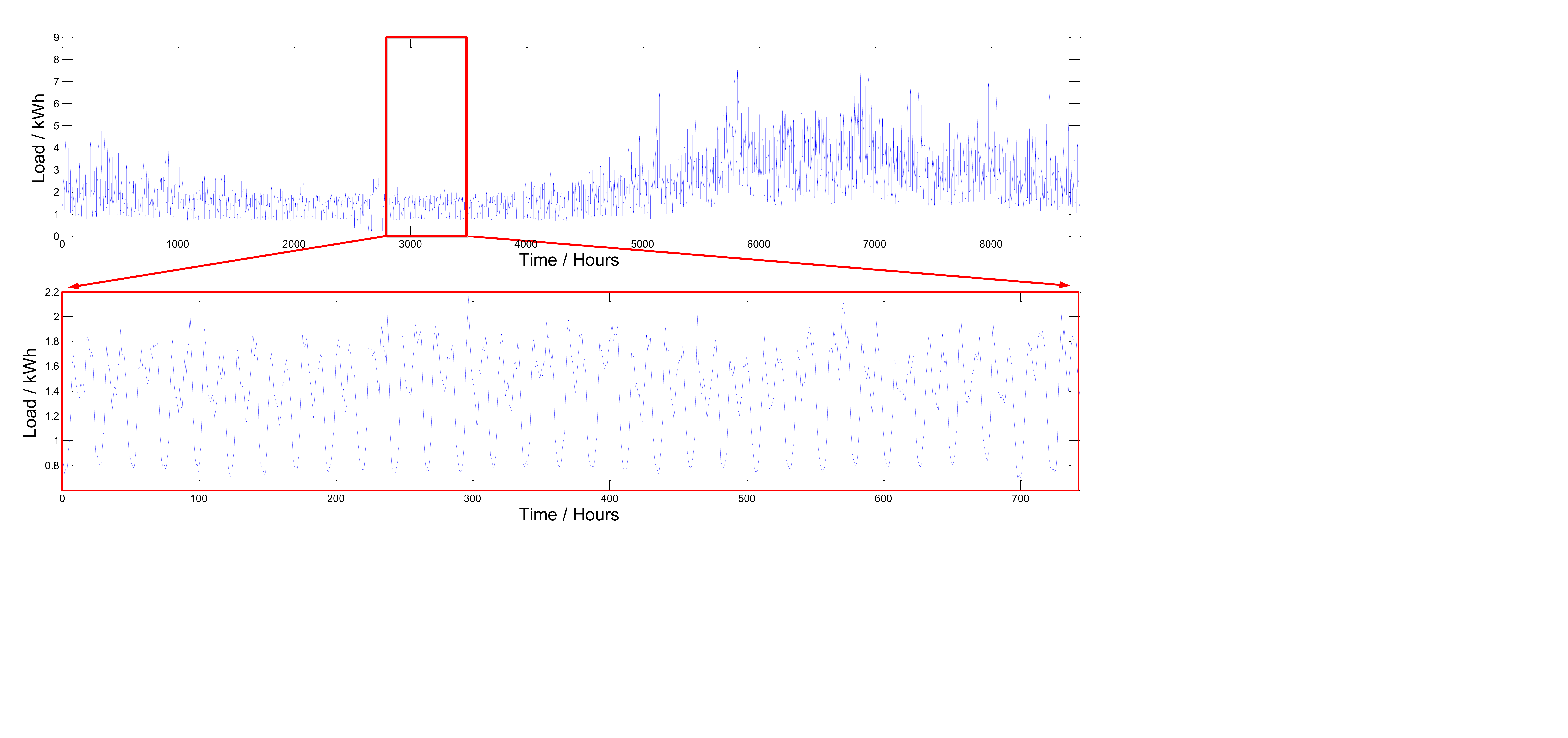}
\caption{Average energy consumption of $112$ residential houses in US for one year from $01/04/2006$ to $31/03/2007$ (above) and data for one month from $01/08/2006$ to $31/08/2006$ (below), provided by Pacific Northwest National Laboratory~\cite{hammerstrom2007pacific}.}\label{fig_timelinePlot}
\vspace{-3.5mm}
\end{minipage}
\hspace{5ex}
\begin{minipage}{3.2in}
\includegraphics[width=3.2in,height = 1.55in]{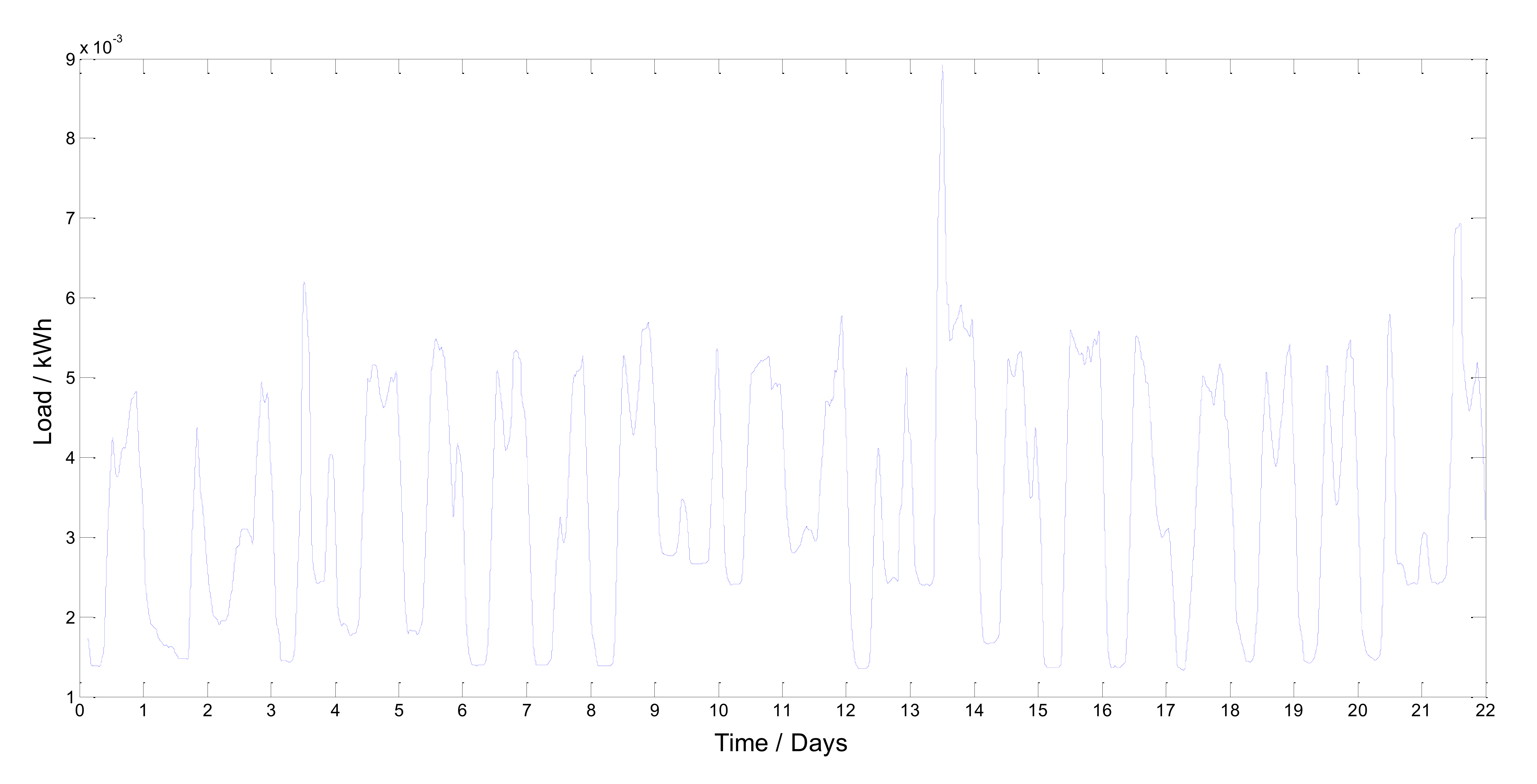}
\caption{Daily energy consumption of an individual residential house in Waterloo, Canada, from $23/01/2011$ to $13/02/2011$, provided by Singh \textit{et al.}, University of Waterloo~\cite{singh2012hourly}.}\label{fig_timelinePlot2}
\end{minipage}
\end{center}
\end{figure*}

\begin{itemize}
\item First, in a short time window (say $1$ to $2$ hours), the correlation between time and the load values may be hard to capture due to two reasons: (1) some random events may play a dominant role in electric load, and (2) it is hard to obtain a unified model to capture the local pattern, which may change over time. 
\item Second, in a relatively long time (say days), even though certain regular patterns of the load curve can be found, the load curve along the timeline is nonlinear and may be too complicated to model with fixed parameters.
\item Third, with landscape data, each sample is usually treated equally, making it difficult to effectively capture special behavioral features. For instance, the energy consumption for a cafeteria is low and stable when it is closed and high during breakfast and lunch times. In this sense, it would be better if load data could be treated differently during the former period (say from 7:00 pm to 7:00 am) and during the latter periods (say from 7:00 am to 9:00 am and from 11:00 am to 1:00 pm).
\end{itemize}

\subsection{Our Contribution}
Based on the above observations, we challenge the traditional landscape data as an efficient way to organize load curve data. The following contributions are made in the paper: 
\begin{itemize}
\item A new view, called {\em portrait}, is proposed for load data analysis. Switching perspective from landscape to portrait, some hidden behavioral patterns in the load data become prominent, such as the numerical stability of load curve data in the same hours of different days.
\item With Fourier analysis, an algorithm is designed to {\em automatically} transform a landscape data to portrait data. We further extend the method to build {\em virtual portrait datasets}, meaning of which will be disclosed later, to address the problem in the third observation raised above.
\item A data pre-processing method is proposed, so that non-stationary load data can be effectively handled with the help of virtual portrait datasets.
\item Efficient algorithms are designed to use virtual portrait data for both small-scale and large-scale load data cleansing. Our experimental results show that our portrait based method is faster and more accurate, compared to the state-of-the-art regression-based methods. 
\end{itemize}

\nop{The rest of the paper is organized as follows. In Section~\ref{sec:relatedwork}, we introduce the related work. In Section~\ref{sec:systemModel}, some concepts about portrait data and its benefits are introduced. We apply Fourier transform to obtain the period of landscape data automatically, and propose algorithms to build (virtual) portrait data from landscape data in Section~\ref{sec:bvPortrait}. In Section~\ref{sec:cleansing}, we design algorithms for load data cleansing, based on virtual portrait data. We further study non-stationary landscape data in Section~\ref{sec:bvLandscape}. We evaluate the performance of our method with small-scale stationary and large-scale non-stationary load data in Section~\ref{sec:evaluation}. The paper is concluded in Section~\ref{sec:conclusion}.}

\section{Related Work}\label{sec:relatedwork}

Load curve data cleansing in smart grid has caught more and more attention recently, from both academia and industry. So far, most related work considers the polluted data as outliers in load pattern and focuses on outlier detection.

Regression-based methods have been widely studied for outlier detection in time series~\cite{chen2010automated, mateos2012robust, ljung1993outlier, abraham1989outlier}. In~\cite{chen2010automated}, a non-parametric regression method based on B-spline and kernel smoothing was proposed and applied to identify polluted data. In~\cite{abraham1989outlier}, the residual pattern from regression models was analyzed and applied to construct outlier indicators, and a four-step procedure for modeling time series in the presence of outliers was also proposed. Greta \textit{et al.}~\cite{ljung1993outlier} considered the estimation and detection of outliers in time series generated by a Gaussian auto-regression moving average (ARMA) process, and showed that the estimation of additive outliers was related to the estimation of missing observations. The ARMA model was also utilized in~\cite{fox1972outliers, abraham1988score, abraham1989outlier, schmid1986multiple} as the basic model for outlier detection. In general, the regression-based methods are established on empirical knowledge and their parameters are regulated manually according to the domain knowledge of experts. As a result, such methods are subject to either underestimation or overestimation.

Since load curve data consists of one-dimensional real values, univariate statistical methods can deal with outliers in such dataset~\cite{ferguson1961rejection, david1979robust, gather1989testing, davies1993identification}. Most univariate methods for outlier detection assume that the data values follow an underlying known distribution. Then, the outlier detection problem is transformed to the problem of finding the observations that lie in a so-called outlier region of the assumed distribution~\cite{davies1993identification}. Even though those methods have been proved simple and effective, we may not always know the underlying distribution of the data. This is unfortunately true for load curve data, \textit{e.g.}, the distribution of the data shown in Figs.~\ref{fig_timelinePlot} and~\ref{fig_timelinePlot2} is unknown.

In addition to the above methods, data mining techniques have also been developed to detect outliers, such as $k$-nearest neighbor~\cite{ramaswamy2000efficient, knox1998algorithms}, $k$-means~\cite{allan1998topic, nairac1999system}, $k$-medoids~\cite{bolton2001unsupervised}, density-based clustering~\cite{kriegel2005density}, \textit{etc.}  In general, these methods classify the observations with similar features, and find the observations that do not belong strongly to any cluster or far from other clusters. Nevertheless, most data mining techniques are designed for structured relational data, which may not align well for the need of outlier detection in load curve data. In addition, these methods are normally time consuming because they need a training process on a large dataset.

\nop{Some of the above methods, especially the regression-based methods, have also been used for load forecasting. Nevertheless, load forecasting and load data cleansing are different applications, and their purposes are different. For load forecasting, as mentioned in~\cite{chen2010automated}, all historical data are trusted and used to forecast the load at a future point in time. For load cleansing, historical records are used to detect corrupted data at a historical point in time, and appropriate values may be needed to replace the corrupted data. Therefore, load cleansing and load forecasting belong to different phases of load analysis: load cleansing comes before load forecasting and provides accurate load information for the latter.}

\section{Introduction of Portrait Data}\label{sec:systemModel}

\nop{We propose a new view of load curve data and organize them via a model of portrait data, which can facilitate the analysis and cleansing of load curve data.}

\subsection{Portrait Data}

\begin{definition}
Consider a periodic function $f(x)$ with period of $T$ defined over $[0, NT]$. We split one period of time $[0, T]$ into $n$ even slices, i.e., $0=x_0 < x_1 < x_2 \ldots < x_n=T$. The \textbf{portrait data} of function $f(x)$ corresponding to the $i$-th time slice ($0 \leq i \le n$), denoted by $p_i$, is defined as the dataset: 
\begin{equation}
p_i := \{ f(x) |x\in \left[x_i+kT, x_{i+1}+kT \right], 0\leq k \leq N\}.
\end{equation}
\end{definition}

\begin{definition} The \textbf{span} of a portrait data $p_i$ is defined as
\begin{equation}
{sp}_i := x_{i+1} - x_i.
\end{equation}
\end{definition}

Similarly, for discrete periodic load curve data with even spacing labeled as $\{y(0),y(1),y(2),\cdots\}$, the portrait data are composed with the data points falling within the corresponding time intervals, \textit{i.e.}, the portrait data $p_i$ is constructed as: 

\begin{equation}\label{Eqt:portraitDef}
p_i := \{ y(t)| t =  t_i+kT, 0\leq k \leq N\}.
\end{equation}

\subsection{Example of Portrait Data }

To help better understand the portrait data, we use the one-month load curve data in Fig.~\ref{fig_timelinePlot} as an example to illustrate portrait data visually.

Noticing that the data exhibits a periodicity of $24$ hours, we divide the original time line by $24$ hours into $31$ slices (days) and re-arrange the slices in parallel. In this way, we transform the 2D landscape data into 3D space, with $x$-axis representing hours, $y$-axis days, and $z$-axis the load values, as shown in Fig.~\ref{fig_portraintPlot1}. To view the energy consumption of each hour in the $31$ slices, we rotate the figure in the $x\mbox{-}y$ coordinate by $90$ degrees, and re-draw the data into $24$ slices. Each slice represents a portrait data consisting of the energy consumption at the same hour in the $31$ days, as shown in Fig.~\ref{fig_portraintPlot2}. Immediately, we can observe that: \emph{the values in each portrait dataset are relatively stable.}

\begin{figure}[!ht]
\begin{center}
\includegraphics[width=2.5in,height=1.4in]{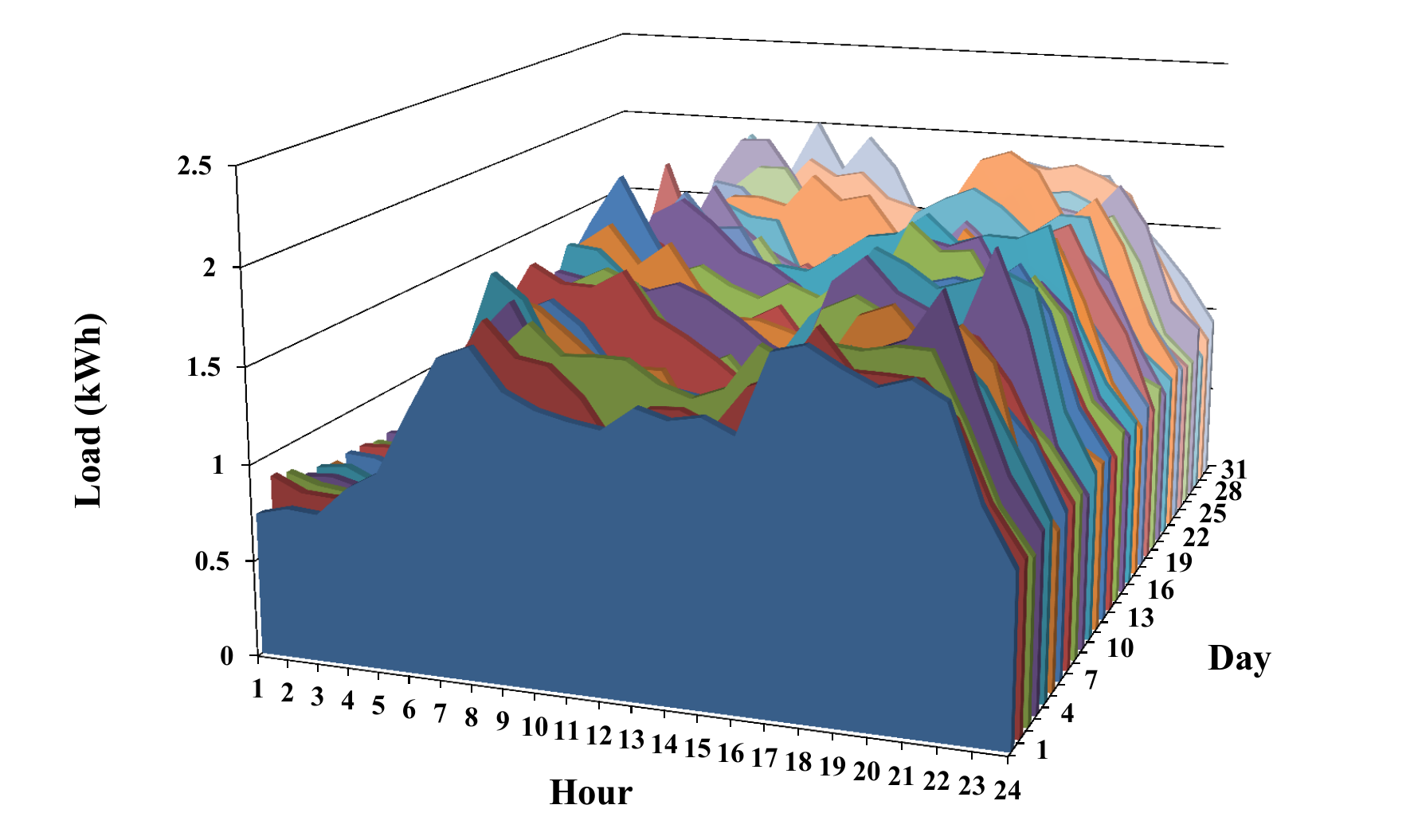}
\caption{Divide timeline into $31$ pieces by $24$ hours and reposition the pieces in parallel}\label{fig_portraintPlot1}
\end{center}
\vspace{-0.1in}
\end{figure}

\begin{figure}[!ht]
\begin{center}
\includegraphics[width=2.5in,height=1.4in]{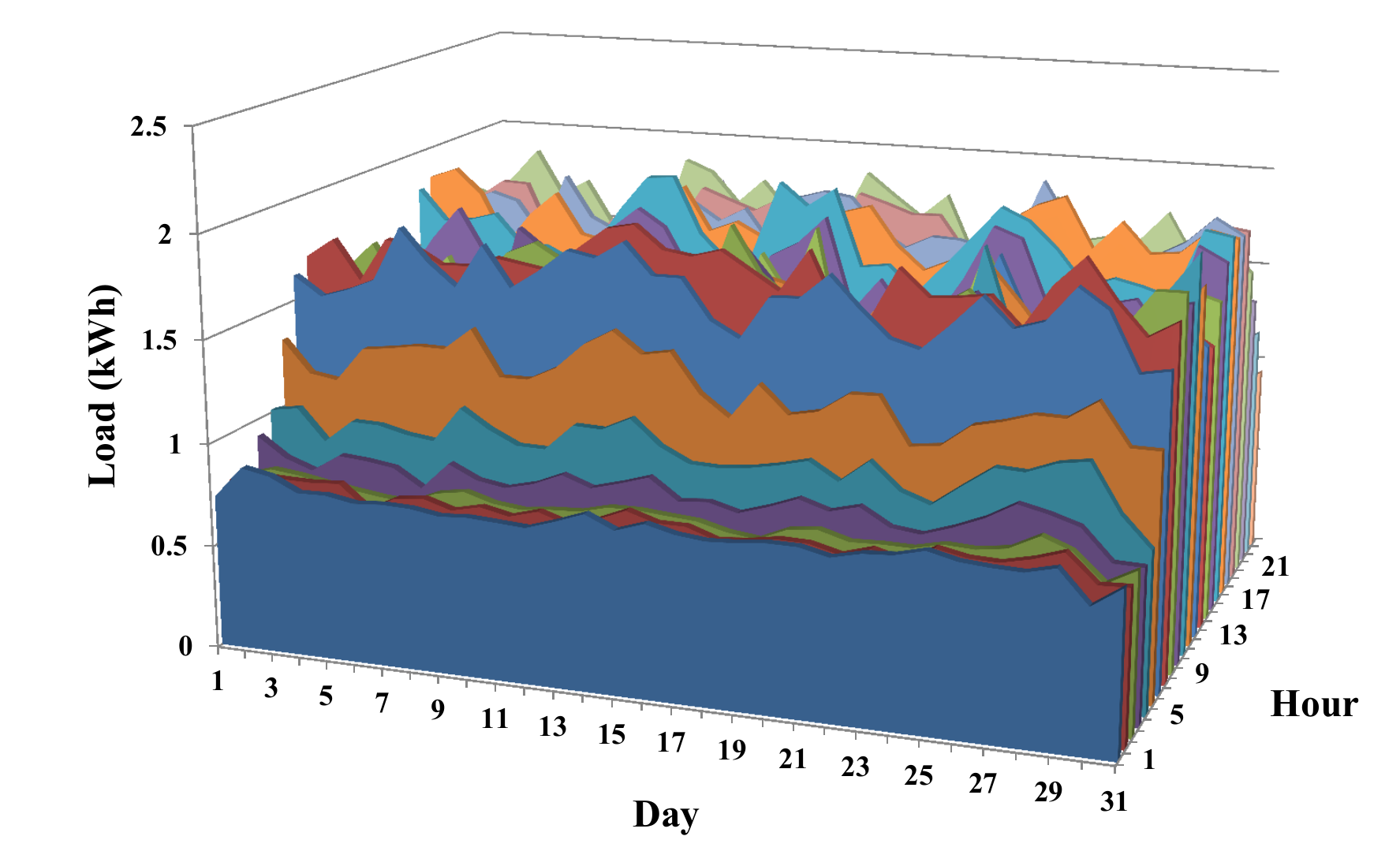}
\caption{Switch the view to portrait}\label{fig_portraintPlot2}
\end{center}
\vspace{-0.1in}
\end{figure}

\subsection{Characteristic Vector of Portrait Data}

Intuitively, a portrait dataset should include values with a very small variation. There are many ways to model this phenomenon. In this paper, we use the following characteristic vector to describe the portrait data.

\begin{definition}  The \textbf{characteristic vector} of portrait data $p_i$ is defined as 
\begin{equation}
e_i := [\theta_i, {M}_i],
\end{equation}
where $\theta_i$ and ${M}_i$ represent the median and the median absolute deviation (MAD) of values in $p_i$, respectively.
\end{definition} 

Since the data may be contaminated by outliers, we use the \emph{median} and \emph{MAD} instead of \emph{mean} and \emph{standard deviation} to represent central tendency and statistical dispersion of a portrait dataset, respectively. The \emph{median} and \emph{MAD} are more robust measures~\cite{huber2011robust}.

For the data in Fig.~\ref{fig_portraintPlot2}, the characteristic vectors of portrait data of the first $10$ hours are summarized in Table~\ref{tab:charvector}. The last column of the table shows that the \emph{MAD} value of landscape data is significantly higher. The results indicate that each portrait dataset is much more stable than the landscape data.\newline

\begin{table}[!ht]
	\caption{The characteristic vectors of portrait data of the first $10$ hours in Fig.~\ref{fig_portraintPlot2}, compared with landscape data (unit: kWh)}
	\centering
		\begin{tabular}{m{0.1cm} m{0.3cm} m{0.3cm} m{0.3cm} m{0.3cm} m{0.3cm} m{0.3cm} m{0.3cm} m{0.3cm} m{0.3cm} m{0.4cm}|m{0.55cm}}
			\hline
		 $hr.$ &  $1$ & $2$ & $3$ & $4$ & $5$ & $6$ &  $7$ & $8$ & $9$ &  $10$ & 1-10\\
			\hline
	 $\theta$ &  $0.79$ & $0.78$ & $0.77$ & $0.84$ & $0.99$ & $1.30$ & $1.69$ & $1.76$ & $1.69$ & $1.60$ & $1.14$\\
	 $M.$ &  $\textbf{0.04}$ & $\textbf{0.01}$ & $\textbf{0.02}$ & $\textbf{0.04}$ & $\textbf{0.05}$ & $\textbf{0.08}$ & $\textbf{0.09}$ & $\textbf{0.07}$ & $\textbf{0.09}$ & $\textbf{0.11}$ & $\textbf{0.42}$\\
			\hline
		\end{tabular}
	\label{tab:charvector}
\end{table}

\begin{definition} The \textbf{similarity} of two portrait datasets $p_i$,$p_j$ with characteristic vectors $e_i$, $e_j$, respectively, is defined as   
\begin{equation}
s_{ij} := \begin{cases} \infty, & \mbox{if } e_i=e_j  \\ 
          1/{\lVert e_i - e_j \rVert}_{2} & otherwise. \end{cases} \label{sij}
\end{equation}
\end{definition}

\nop{Applying different threshold values on the similarity, heuristic algorithms can be developed to obtain different number of portrait data slices. We can also build virtual portrait datasets for ease of data analysis, which will be introduced in detail in Section~\ref{sec:bvPortrait}.}

We can develop heuristic algorithms to merge multiple portrait datasets with a high similarity into a virtual portrait dataset, which will be introduced in detail in Section~\ref{sec:bvPortrait}.

\subsection{Properties of Portrait Data}

Compared to landscape data, the portrait data has the following desirable properties:
\begin{itemize}
\item The data values within the same portrait are similar and can be processed together even if they are separated in the original time domain. 

\item The data values within the same portrait dataset can be captured with a simple model, for which numerous fast data cleansing methods can be applied. In contrast, landscape data is normally nonlinear and requires complicated non-linear regression-based methods. 

\item With portrait data, users' behavioral patterns in different time periods can be modeled. In Fig.~\ref{fig_portraintPlot2}, the energy consumption in the first hour of each day is quite stable and low, but the situation for the seventh hour is quite different. As such, a data point with small deviation in the first slice should be captured as an outlier, but may be considered regular in the seventh slice. In this way, we can improve the accuracy of outlier detection.
\end{itemize}

It is worth noting that portrait data is \textit{not} just a data visualization trick. It is helpful to design efficient algorithms for load curve data analysis and cleansing. In specific, due to the stability in each portrait data, it is much easier to build simple models to capture the outliers. In addition, by combining similar portrait slices into one virtual slice, we can build \emph{virtual portrait} dataset, which further speeds up data processing.

\section{Construction of Portrait Data}\label{sec:bvPortrait}

\subsection{Analysis of Periodic Pattern in Landscape Data}\label{sec:fftAnalysis}
 
In order to \emph{automatically} construct portrait data, we need to find out the time period of landscape load curve data. In our daily life, the energy consumption of different houses or buildings is usually periodic, either hourly, daily or weekly. When the volume of landscape data is big, an automatic method is needed to quickly discover the periodic behaviour hidden in the landscape data. In this paper, Fourier analysis~\cite{bloomfield2004fourier} is used for this purpose.

According to Fourier transform, given a non-sinusoidal periodic function:
\begin{equation}
f(t)=f(t+kT), k=0, 1, 2, \cdots,
\end{equation}
if in one cycle of the periodic function there are finite maximum and minimum values, as well as the finite number of first category discontinuous points\footnote{A discontinuous point $x$ is called the first category discontinuous point where there exist finite limits from the left $f(x-0)$  and from the right $f(x+0)$ for $f$.}, the function can be unfolded into a convergent Fourier Series, \emph{i.e.},

\begin{equation}\label{eq_fft}
f(t) = A_{0} + A_{1}cos(\Omega t + \psi_{1}) + \sum_{k=2}^{\infty}A_{k}cos(k \Omega t + \psi_{k}),
\end{equation}
where $A_{0}$ is called the \emph{constant component} and $A_{1}cos(\Omega t + \psi_{1})$ the \emph{fundamental component}. The frequency of the fundamental component discloses the lowest frequency in the original function $f(t)$, which can be used to construct the portrait data. 

\nop{According to Fourier transform, given a non-sinusoidal periodic function, if in one cycle of the periodic function there are finite maximum and minimum values, as well as the finite number of first category discontinuous points\footnote{A discontinuous point $x$ is called the first category discontinuous point where there exist finite limits from the left $f(x-0)$ and from the right $f(x+0)$ for $f$.}, it can be unfolded into a convergent Fourier Series. In the Fourier Series, the second formula coefficient, named \emph{fundamental component}, discloses the period of the function.}

Since the load curve data is discrete, we should use another form for Fourier transform, Discrete Fourier Transform (DFT), to convert a finite list of equally-spaced samples of a function into the list of coefficients of a finite combination of complex sinusoids, ordered by their frequencies. To speed up the process, Fast Fourier Transform (FFT) is adopted, which is developed upon DFT and works much faster. 

In practice, the sampling interval for residential energy consumption on the utility side is normally $15$ minutes~\cite{chen2011activity}. Considering the periodic pattern in load curve is relatively longer, such as one day (24 hours), the sampling rate is high enough to acquire the time period of load curve data.

\nop{In our application context, sampling rate for energy consumption is high enough (normally, 15 minutes in practice~\cite{chen2011activity}) to recover the periodic information in the load curve data. Some recent products~\cite{CurrentCost2013} can sample and transmit data at a rate as high as per sample every 10 seconds. Thus, the acquisition of time period can be guaranteed.}

\subsection{Construction of Basic and Virtual Portrait Data}

\nop{Having identified the periods in landscape data, now the problem is how many slices of portrait data should be split. }
The next step is to decide how many slices of portrait data should be split.

One solution is to split the load curve data with the span of sampling interval, which will result in portrait data with the highest resolution. However, since the sample rate may be significantly high, such kind of splitting may result in too many portrait data slices. Considering that the characteristic vectors of some portrait datasets are similar, we merge them together into a virtual portrait dataset to speed up data cleansing. Therefore, a two-phase method is developed.

\subsubsection{Build Basic Portrait Datasets} The portrait datasets obtained in this phase are called \emph{basic portrait dataset} (BPD). With FFT, the fundamental period of load curve data can be obtained. Assuming that there are $r$ samples in one period, we can obtain $r$ basic portrait datasets $\{p_0, p_1,\cdots, p_r\}$. Accordingly, we can calculate the characteristic vector of each basic portrait dataset, denoted by $\{e_0, e_1,\cdots, e_r\}$, respectively. 

\subsubsection{Build Virtual Portrait Datasets} We merge multiple basic portrait datasets with similar characteristic vectors into one \emph{virtual portrait dataset} (VPD). As such, A clustering algorithm is needed to partition the basic portrait datasets into exclusive clusters such that within each cluster, the pairwise similarity of basic portrait datasets is no less than a given threshold. In order to accelerate data analysis, it is desirable to minimize the total number of clusters. This optimization problem can be formulated as follow:
\begin{itemize}
\item \textbf{Input:} Basic portrait data $\{p_1,p_2,\cdots,p_r\}$ and their corresponding characteristic vectors $\{e_1,e_2,\cdots,e_r\}$. A given threshold $s_0$ on similarity.
\item \textbf{Output:} Minimum number of virtual portrait datasets, denoted by $\{P_1,P_2,\cdots,P_n\}$ such that within each virtual portrait dataset, the pairwise similarity of the basic portrait datasets is no less than $s_0$. 
\end{itemize}

\begin{equation}\label{eqt:optPortraitData}
\begin{aligned}
 \underset{\{P_1,P_2,\cdots,P_n\}}{\text{\textbf{  minimize}}} &\quad \quad \quad  n \\
 \text{\textbf{subject to}} && \\
& \bigcup P_i=\{p_1,p_2,\cdots,p_r\}\\
& P_i\bigcap P_j=\varnothing, i\neq j\\
& P_i=\left\{\{p_{l_1},p_{l_2},\cdots,p_{l_m}\} \ |\ s_{l_sl_t} \geq s_0 \right\}\\
& 1\leq i,j\leq n; 1\leq m\leq r; 1\leq l_s,l_t\leq m
\end{aligned}
\end{equation}

In order to solve the above problem, a graph $G=(V,E)$ is constructed, where each vertex $v\in V$ represents a BPD and an edge is built between two vertices if their similarity is no less than $s_0$. It is easy to see that the problem is equivalent to the \emph{clique-covering} problem, which has been proven to be NP-complete~\cite{lund1994hardness}. Hence, a \emph{greedy clique-covering algorithm} is adopted to obtain an approximate solution. Algorithm~\ref{alg:greedyCovering} shows the pseudo code of the greedy clique-covering problem. 

\begin{algorithm}[!ht]
\caption{Greedy Clique-Covering Algorithm}\label{alg:greedyCovering}
\begin{algorithmic}[1]
\Require Graph $G=(V,E)$
\Ensure A set of cliques $P$ that completely cover $G$

\State Initialize uncovered vertex set $V' \leftarrow V$
\State Initialize number of cliques, $n=0$

	\While{$V' \neq \Phi$}
	\State $n= n+1$
	\State Find $v\in V'$ with the highest node degree
	\State Find $U\subseteq V'$ with $u\in U$ and $(u,v)\in E$
	\State Construct subgraph $G'=(U,D)$ where $U$ includes all vertices adjacent to $v$, and $D$ includes the associated links
	\State Initialize clique $P_n=\{v\}$

		\For {each $w \in U$} 
			\If {$w$ is adjacent to all vertices in $P_n$}
			\State $P_n\leftarrow P_n \cup \{w\}$
			\EndIf
			
		\EndFor

	\State $V'\leftarrow V'\backslash P_n$
	\EndWhile \\

\Return $P_1,P_2,\cdots,P_n$

\end{algorithmic}
\end{algorithm}

The basic idea of the algorithm is to find cliques that cover more vertices that have not been clustered. Heuristically, the vertices with larger degrees may have a better chance of resulting in a smaller number of cliques. Thus, the search starts from the vertex with the largest degree, until all vertices are covered. Obviously, a resulted cluster is a clique in the graph. Since each vertex represents a BPD, a clique represents a VPD.

\begin{lemma}
The computational complexity of Algorithm~\ref{alg:greedyCovering} is lower bounded by $O(r\log{r})$ and upper bounded by $O(r^{2}\log{r})$, where $r$ is the number of basic portrait datasets. 
\end{lemma}

\begin{proof}
Since the similarity of two basic portrait datasets is calculated with their characteristic vectors consisting of two values, the graph $G$ in Algorithm~\ref{alg:greedyCovering} is actually a geometric graph in the 2D plane. Any clique resulted from Algorithm~\ref{alg:greedyCovering} can be bounded by some rectangle region in the 2D plane. According to~\cite{imai1983finding},  the largest clique of a rectangle intersection graph can be found with computational complexity no more than $O(r\log{r})$. Since in Algorithm~\ref{alg:greedyCovering} the number of iterations in finding cliques could range from $1$ to $r$, the computational complexity ranges from $O(r\log{r})$ to $O(r^{2}\log{r})$.
\end{proof}

\section{Load Curve Data Cleansing}\label{sec:cleansing}

In this section, we show portrait data can help load curve data cleansing. Load curve data cleansing involves two phases: (1) detecting outliers and (2) fixing the missing or aberrant values in the dataset. 

\subsection{Detection of Outliers}
Formally, for a given distribution $F$, the outlier detection problem is to identify those values that lie in a so-called \emph{outlier region} defined below:

\begin{definition} For any confidence coefficient $\alpha, 0<\alpha<1$, the \textbf{$\alpha$-outlier region} of $F$ distribution with parameter vector $\Theta$ can be defined by
\begin{equation}\label{eq:outlierRegion}
out(\alpha,\Theta)=\{x: x<Q_{\frac{\alpha}{2}}(\Theta) \textbf{ or } x>Q_{1-\frac{\alpha}{2}}(\Theta)\},
\end{equation}
\noindent where $Q_{q}(\Theta)$ is the $q$ quantile of function $F(\Theta)$.
\end{definition}

Since we usually do not have \textit{apriori} knowledge on the distribution of portrait data, various possible cases should be considered. Note that performing statistical test to find out the distribution of load curve data does not work well when the load data is polluted. We need to consider several potential cases for outlier detection. 
 
\subsubsection{Case 1: Outlier Detection for Normal Distributed Data}

The normal distribution can be adopted as an empirical distribution, which has been proved to be effective in general situations~\cite{ben2005outlier}. 

According to Equation~(\ref{eq:outlierRegion}), for a normal distribution $N(\mu,\sigma^2)$, its  $\alpha$-\emph{outlier region} is

\begin{equation}
out(\alpha,(\mu,\sigma^2))=\{x: |x-\mu|>\Phi_{1-\frac{\alpha}{2}}\sigma\},
\end{equation}
\noindent where $\Phi_{q}$ is the $q$ quantile of $N(0,1)$. For normal distributed portrait datasets $P_i,i=1,2,\cdots$, we claim that a value $x$ is an $\alpha$-outlier in $P_i$, if $x\in out(\alpha,(\widehat{\mu_i},\widehat{\sigma_i}^2))$, where $\widehat{\mu_i}$ and $\widehat{\sigma_i}$ are unbiased estimators of $\mu_i$ and $\sigma_i$, respectively. Since the data may be contaminated by outliers, we use the \emph{median} and \emph{MAD} instead of \emph{mean} and \emph{standard deviation} in our later detection.  

\nop{For our virtual portrait datasets, $P_i,i=1,2,\cdots$, assuming the data points follow a normal distribution $N(\mu_i,\sigma_i^2)$, we can use the above model for outlier detection. In its \emph{characteristic vector} $e_i=[\widehat{\mu_i},\widehat{\sigma_i}]$, since $\widehat{\mu_i}$ and $\widehat{\sigma_i}$ can be regarded as the unbiased estimators of $\mu_i$ and $\sigma_i$, the original distribution can be proximately replaced by $N(\widehat{\mu_i},\widehat{\sigma_i}^2)$, especially when the size of portrait datasets is large. }

\subsubsection{Case 2: Outlier Detection for Gamma Distributed Data}

It has been shown that the aggregated residential load at a given time instant follows the gamma distribution~\cite{carpaneto2008probabilistic, cagni2004characterisation}. \nop{In addition, for the unpolluted portrait data in our paper, we perform statistical tests, including the Kolmogorov-Smirnov (KS) test and chi-square test, and find that it also fits the gamma distribution well.} In this light, the gamma distribution is also a good candidate distribution for outlier detection. 

According to Equation~(\ref{eq:outlierRegion}), for a gamma distribution with shape parameter $\beta$ and scale parameter $\gamma$, $G(\beta,\gamma)$, its  $\alpha$-\emph{outlier region} is
\begin{equation}
out(\alpha,(\beta,\gamma))=\{x: x<F_{\frac{\alpha}{2}}^{-1}(\beta,\gamma) \textbf{ or } x>F_{1-\frac{\alpha}{2}}^{-1}(\beta,\gamma) \}, \label{GammaDetection}
\end{equation}
\noindent where $F^{-1}$ is the inverse cumulative distribution function of $G(\beta,\gamma)$, and $F_{q}^{-1}(\beta,\gamma)$ is the $q$ quantile of $G(\beta,\gamma)$.

If we assume that virtual portrait datasets, $P_i,i=1,2,\cdots$ follow a gamma distribution $G(\beta,\gamma)$, we can use (\ref{GammaDetection}) for outlier detection. In this case, $(\widehat{\mu_i}^2/\widehat{\sigma_i}^2)$ and $(\widehat{\sigma_i}^2/\widehat{\mu_i})$ are the moment estimators of $\beta$ and $\gamma$, respectively. 
\nop{When the number of points in the portrait dataset is large enough, the distribution can be proximately represented by $G(\widehat{\mu_i}^2/\widehat{\sigma_i}^2,\widehat{\sigma_i}^2/\widehat{\mu_i})$. } Due to the same reason as in Case 1, we use the \emph{median} and \emph{MAD} instead of \emph{mean} and \emph{standard deviation} in our later detection. 

\nop{Therefore, for gamma distributed portrait data, we claim that a sample data $x$ is an $\alpha$-outlier in $P_i$, if $x\in out(\alpha,\widehat{\mu_i}^2/\widehat{\sigma_i}^2,\widehat{\sigma_i}^2/\widehat{\mu_i})$, where $(\widehat{\mu_i}^2/\widehat{\sigma_i}^2)$ and $(\widehat{\sigma_i}^2/\widehat{\mu_i})$ are moment estimators of $\beta$ and $\gamma$, respectively.}

\subsubsection{Case 3: Outlier Detection for Small-Size Portrait Data}

In the above outlier detection strategies, the size of portrait datasets is assumed to be large. Otherwise, the parameter estimation may be inaccurate. When the size of samples is small, Tukey \textit{et al.}~\cite{tukey1977exploratory} introduce a graphical procedure called \emph{boxplot} to summarize univariate data.

The boxplot uses median and lower and upper quartiles (defined as the $25$-th and $75$-th percentiles). If the lower quartile is $Q_1$ and the upper quartile is $Q_3$, then the difference $(Q_3 - Q_1)$ is called interquartile range or $IQR$. After arranging data in order, the ones falling in the following outlier region are identified as outliers.

\begin{equation}
out(\rho,(Q_1,Q_3))=\{x: x < Q_1-\rho\cdot IQR \textbf{ or } x > Q_3+\rho\cdot IQR \},
\end{equation}
\noindent where $\rho$ is an index of significance, and the outliers are said to be ``mild'' when $\rho=1.5$ and ``extreme'' when $\rho=3$.

Overall, the above three cases cover most situations a user may meet in portrait load data cleansing. Nevertheless, other strategies can also be chosen as long as they give more precise model of the portrait data.

\subsection{Replacing Missing Data or Aberrant Data} \label{sec:imputation}
We mainly focus on outlier detection in this paper for two reasons. 
\begin{enumerate}
	\item Imputation of missing data can be easily done after we obtain the characteristic vector of portrait data, \emph{e.g.}, we could replace a missing value with the median of the corresponding portrait dataset. 
	\item Replacing aberrant values requires human interaction, since it needs the user to further confirm whether or not an outlier is a corrupted value. The user can either (1) replace the outlier with an acceptable value of the corresponding dataset, \emph{e.g.}, the mean value for Case 1 and Case 3 and the value of $\beta/\gamma$ for Case 2, or (2) leave the outlier unchanged, as long as the cause of creating the outlier can be explained, such as the stimulation of holidays/special events.  
\end{enumerate}

Note that data imputation is normally carried out after outlier detection. It is common to initially set the missing data to default values of zeros, which are likely to be outliers and then are replaced with acceptable values. This strategy has been used in~\cite{chen2011activity}. Nevertheless, the default value can be altered by the user according to different scenarios. For instance, if there exist valid load values close to zero, we can set the default value for missing data to a very large value, so that missing data can be identified easily as outliers. 

\nop{Note that data imputation is normally carried out after outlier detection. It is common to initially set the missing data to default values of zeros, which are likely to be outliers and then are replaced with acceptable values. This strategy has been used in~\cite{chen2011activity}. Nevertheless, the default value can be altered by the user according to the overall energy consumption, \textit{e.g.}, if there exist small load curve data close to zero, the default value of missing data can be set to some tremendous values to avoid ``misdiagnosis''.}

If the user has explicitly learned the cause of the aberrant or missing data and found any above or other replacing approach that fits the needs well, the approach can be incorporated with our outlier detection approach and works automatically.

\section{Handling Non-Stationary Landscape Data}\label{sec:bvLandscape}

The construction and cleansing of portrait data are based on the assumption that the landscape data are stationary along the timeline. Informally, a stationary time series has a well-defined mean around which it can fluctuate with constant finite variance. This assumption may be true during a short time period, while in a significantly long time period, such as one year, the load curve shows seasonal patterns and is usually not stationary. For the one year load values shown in Fig.~\ref{fig_timelinePlot}, they are not stationary along the whole timeline. Consequently, the portrait data cleansing strategies in Section~\ref{sec:cleansing} may not work well. 

To deal with this problem, a pre-processing method is proposed, based on two observations: (1) the load curve data exhibits periodicity in a small time scale (\textit{e.g.}, one month) and (2) the fundamental period of load curve data (\textit{e.g.}, one day) in different small time windows (\textit{e.g.}, January and June) is (nearly) the same. Both observations will be further validated in our late experimental evaluation.

Within a small time-scale (\textit{e.g.}, several days to one month), the fundamental period of the landscape data can be obtained via FFT. We first divide the whole time with the length of the fundamental period, and use the data within each time period as the basic building block. For the landscape data in the $i$-th period, which is denoted as $l_{i}$, we define its \emph{characteristic vector} $e_{i} := [\theta_i, {M}_i]$, where $\theta_{i}$ and ${M}_{i}$ represent \emph{median} and \emph{median absolute deviation (MAD)} of the values in $l_{i}$, respectively. Thus, similar to portrait data, the \emph{similarity} of the landscape data in two different period, $l_{i}$ and $l_{j}$, can be defined as $s_{ij} := 1/\lVert e_{i} - e_{j} \rVert$, or $\infty$ if $\lVert e_{i} - e_{j} \rVert =0$. Here we slightly abuse the notation by using the $e_i$ and $s_{ij}$ to denote the characteristic vector and similarity, respectively, for both landscape data and portrait data. Their meaning, however, is easy to figure out from the context. 

For the landscape data of different periods that have similar characteristic vectors, we merge them into one dataset, which is called a \emph{virtual landscape dataset} (VLD). If the whole landscape data consists of $n$ (non-overlapping) periods, the problem of constructing its VLDs can be formally defined as follow.

\begin{itemize}
\item \textbf{Input:} Landscape data $\{ l_{1}, l_{2}, \cdots, l_{n}\}$ and their characteristic vectors $\{ e_{1}, e_{2}, \cdots, e_{n}\}$. A given similarity threshold $s_{0}$.
\item \textbf{Output:} Minimum number of VLDs $\{ L_{1}, L_{2}, \cdots, L_{m}\}$, $m \ll {n}$.
\end{itemize}

Note that the above problem is exactly the same with Problem~(\ref{eqt:optPortraitData}). Thus, Algorithm~\ref{alg:greedyCovering} can be re-used to construct VLDs. Since all data points in each of the VLDs have similar properties, they are stationary and meet the requirement for portrait data construction and cleansing. We can then further build corresponding portrait data for each VLD.

\section{Implementation and Experimental Evaluation}\label{sec:evaluation}

In this section, the real-world trace data shown in Fig.~\ref{fig_timelinePlot} is used to construct virtual portrait datasets. We implement the multiple strategies introduced above to detect outliers, and perform numerous experiments to evaluate the performance. 

So far, there is limited literature on data cleansing applications to smart grid, and one significant contribution was in~\cite{chen2010automated}, in which a non-parametric regression method based on B-spline smoothing was proposed to help users identify outliers. For comparison purpose, we implement the B-spline smoothing method and compare it with our own method.

\subsection{Fundamental Period}

By applying FFT on the landscape data, we got the frequency spectrum of landscape data, and the frequency of the second peak corresponds to the fundamental frequency. Its reciprocal is the fundamental period of the landscape data. After calculation, the fundamental frequency of the landscape data in Fig.~\ref{fig_timelinePlot} is $1.1574 \times 10^{-5}$, which precisely results in a period of $24$ hours ($86400$ seconds).

In addition, a sensitivity experiment is made with $1000$ tests on the data. In each of the test, a random time period longer than $1$ month but shorter than $1$ year was chosen. According to the results, the mean value of the identified periods is $23.9984$ hours with variance of $1.9952 \times 10^{-4}$. Therefore, we can conclude that the accuracy of identified fundamental period is not sensitive to the time period and the starting time of the samples. 

\nop{This experiment also indicates that, by applying FFT, we can get the period of load curve data accurately by using a small number of samples.}

\subsection{The Optimal Threshold Value}

By applying Algorithm~\ref{alg:greedyCovering}, a number of virtual (portrait/landscape) datasets can be built for a given threshold value on the similarity measure. By changing the threshold value from small to large, we can get a series number of virtual datasets. We are thus faced with the following question: what is the optimal threshold value?

In order to answer the above question, \emph{mean distance} is defined to estimate the ``quality" of virtual datasets (\textit{i.e.}, whether or not two virtual datasets are clearly separate). For $n$ virtual datasets with corresponding characteristic vectors $e_1, e_2, \ldots, e_n$, the mean distance is defined as:
\begin{equation}\label{eq:clarity}
d_{n} := {\sum_{i=1}^{n-1}{\sum_{j=i+1}^{n}{s_{ij}}}}/{\binom{n}{2}},
\end{equation}
\noindent where $s_{ij}$ is defined by Equation~(\ref{sij}). Obviously, with the same number of virtual datasets, the larger the $d$, the clearer the separation among the virtual datasets.  

By changing the threshold value on the similarity measure, we can obtain different numbers of virtual datasets. Applying the \emph{ELBOW criterion}\footnote{The concept of VPD is in principle the same as clustering. The ELBOW criterion means that we should choose a number of clusters so that adding another cluster would not model the data much better.}~\cite{Elbow2013}, we can get the optimal number of virtual datasets. The optimal threshold on the similarity measure is thus the one that leads to this number of virtual datasets.

\subsection{Performance Metrics}

In outlier detection, four statistical indicators are widely used: (1) true positive ($TP$), the number of points that are identified correctly as outliers; (2) false positive ($FP$), the number of points that are normal but are identified as outliers; (3) true negative ($TN$), the number of points that are normal and are not identified as outliers; (4) false negative ($FN$), the number of points that are outliers but are not identified. Using $TP, FP, TN$ and $FN$, we evaluate the following four broadly-used performance metrics: accuracy, precision, recall, and F-measure. Accuracy is the degree of closeness of measurement to the actual situation as a whole; precision is the percentage of correctly detected corrupted regions with regard to the total detected regions; recall is the percentage of correctly detected regions with regard to pre-labeled corrupted regions; the F-measure is a harmonic mean of precision and recall, \textit{i.e.},
\begin{equation}
\textit{F-measure}=\frac{2\cdot Precision\cdot Recall}{Precision+Recall}.
\end{equation}

Furthermore, running time (R.T.) and memory usage (M.U.) of program are used to measure the time and space consumption of different methods, respectively. We implement them in R and test them with $32$-bit Windows OS with $3.4 GHz$ CPU and $4 GB$ RAM.

The real-world data shown in Fig.~\ref{fig_timelinePlot} is used for the evaluation. Since this dataset is relatively clean, we ask three students to distort the data with ``falsification'', \textit{i.e.}, they are asked to arbitrarily modify the load curve data within the range of $[0,\infty)$. Five percent of the samples are changed and labelled.

\nop{
\begin{table*}[!ht]
	\caption{Performance on small-scale data: virtual portrait data cleansing \textsl{vs} B-spline smoothing}
	\centering
		\begin{tabular}{|l|c|c|c||c|c|c|c|c|c|}
			\hline
			& \multicolumn{3}{c||}{\textbf{Virtual Portrait Data Cleansing Strategies}} & \multicolumn{6}{c|}{\textbf{B-spline Smoothing}}\\
      \cline{2-10}
			 & Normal-based & Gamma-based & IQR-based & df = $148$ & df = $188$ & df = $228$ & df = $258$ & df = $288$ & df = $318$ \\
			\hline
			$Accuracy$ & $0.9878$ & $0.9865$ & $0.9879$ & $0.9582$ & $0.9716$ & $0.9715$ & $0.9724$ & $0.9779$ & $0.9748$ \\
			\hline
			$Precision$ & $0.8857$ & $0.9375$ & $0.9118$ & $0.8182$ & $0.7917$ & $0.7308$ & $0.5405$ & $0.5333$ & $0.4151$ \\
			\hline
			$Recall$ & $0.7750$ & $0.7500$ & $0.7750$ & $0.2250$ & $0.4750$ & $0.4750$ & $0.5000$ & $0.6000$ & $0.5500$ \\
			\hline
			$F-measure$ & $0.8267$ & $0.8333$ & $0.8378$ & $0.3529$ & $0.5938$ & $0.5758$ & $0.5195$ & $0.5647$ & $0.4731$ \\
			\hline
			\hline
			$Running Time(s)$ & $0.036$ & $0.036$ & $0.037$ & $0.040$ & $0.051$ & $0.075$ & $0.089$ & $0.101$ & $0.118$ \\
			\hline
			$Memory Usage(MB)$ & $0.067$ & $0.067$ & $0.067$ & $2.933	$ & $3.605$ & $4.327$ & $4.890$ & $5.422$ & $5.980$ \\
			\hline
		\end{tabular}
	\label{tab:resultComparison}
\end{table*}
}

\begin{table*}[!ht]
	\caption{Performance on small-scale data: virtual portrait data cleansing \textsl{vs} B-spline smoothing}
	\centering
		\begin{tabular}{|l|c|c|c|c||c|c|c|c|c|}
			\hline
			& \multicolumn{4}{c||}{\textbf{Virtual Portrait Data Cleansing Strategies}} & \multicolumn{5}{c|}{\textbf{B-spline Smoothing}}\\
      \cline{2-10}
			 & Normal-based & Gamma-based & IQR-based & B-spline(df = $68$) & df = $148$ & df = $188$ & df = $228$ & df = $258$  & df = $318$ \\
			\hline
			$Accuracy$ & $0.9878$ & $0.9865$ & $0.9879$ & $0.9823$ & $0.9582$ & $0.9716$ & $0.9715$ & $0.9724$  & $0.9748$ \\
			\hline
			$Precision$ & $0.8857$ & $0.9375$ & $0.9118$ & $0.7500$ & $0.8182$ & $0.7917$ & $0.7308$ & $0.5405$  & $0.4151$ \\
			\hline
			$Recall$ & $0.7750$ & $0.7500$ & $0.7750$ & $0.6750$ & $0.2250$ & $0.4750$ & $0.4750$ & $0.5000$  & $0.5500$ \\
			\hline
			\textit{F-measure} & $0.8267$ & $0.8333$ & $0.8378$ & $0.7105$ & $0.3529$ & $0.5938$ & $0.5758$ & $0.5195$  & $0.4731$ \\
			\hline
			\hline
			$R.T.(second)$ & $0.036$ & $0.036$ & $0.037$ & $0.021$ & $0.040$ & $0.051$ & $0.075$ & $0.089$  & $0.118$ \\
			\hline
			$M.U.(MB)$ & $0.067$ & $0.067$ & $0.067$ & $1.427$ & $2.933	$ & $3.605$ & $4.327$ & $4.890$  & $5.980$ \\
			\hline
		\end{tabular}
	\label{tab:resultComparison}
\end{table*}

\begin{figure*}[!ht]
\begin{center}
\includegraphics[width=1.0\textwidth]{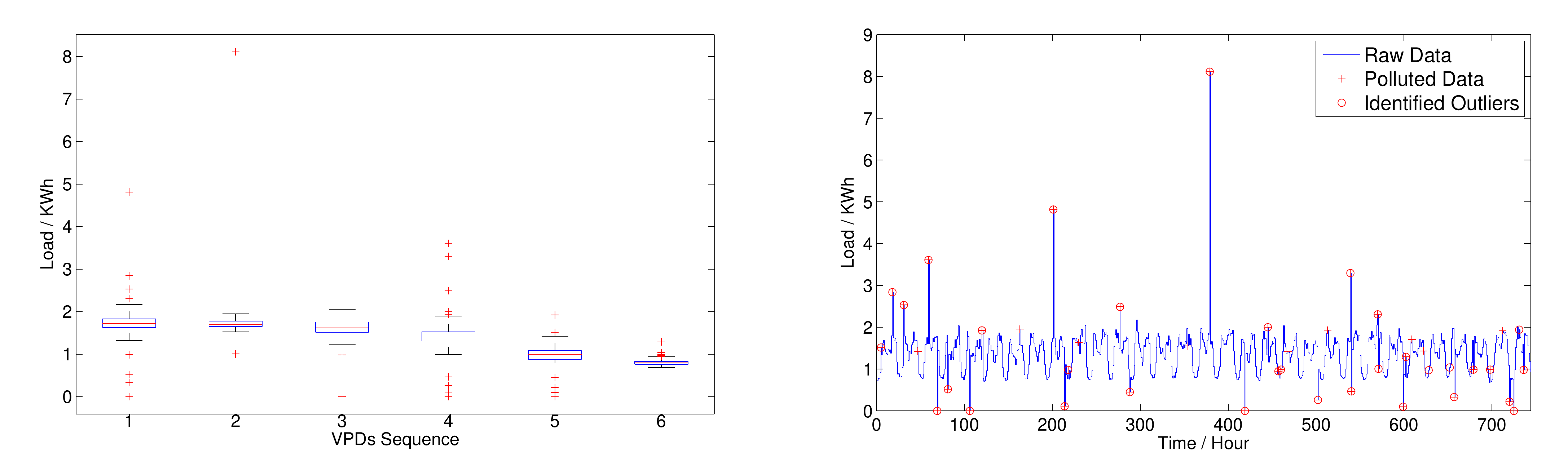}
\caption{Result of outlier detection from IQR-based virtual portrait data cleansing}\label{fig_resultBoxplot}
\end{center}
\end{figure*}

\begin{table*}[!ht]
	\caption{Performance on large-scale data: virtual portrait data cleansing \textsl{vs} B-spline smoothing}
	\centering
		\begin{tabular}{|l|c|c|c||c|c|c||c|c|c||c|c|}
			\hline
			& \multicolumn{3}{c||}{\textbf{5 Virtual Landscape Datasets}} & \multicolumn{3}{c||}{\textbf{7 Virtual Landscape Datasets}} & \multicolumn{3}{c||}{\textbf{10 Virtual Landscape Datasets}} & \multicolumn{2}{c|}{\textbf{B-spline Smoothing}}\\
      \cline{2-12}
			 & N-based & G-based & I-based & N-based & G-based & I-based & N-based & G-based & I-based & df = $2815$ & df $>$ $2815$ \\
			\hline
			$Accuracy$ & $0.9931$ & $0.9950$ & $0.9927$ &   $0.9939$ & $0.9950$ & $0.9930$ &   $0.9936$ & $0.9951$ & $0.9947$ &   $0.9792$ & $-$\\
			\hline
			$Precision$ & $0.6154$ & $0.7080$ & $0.7143$ &   $0.5820$ & $0.7080$ & $0.7590$ &   $0.6939$ & $0.6864$ & $0.7091$ &   $0.3378$ & $-$\\
			\hline
			$Recall$ & $0.5161$ & $0.6452$ & $0.4839$ &   $0.5726$ & $0.6452$ & $0.5081$ &   $0.5484$ & $0.6532$ & $0.6290$ &   $0.3620$ & $-$\\
			\hline
			\textit{F-measure} & $0.5614$ & $0.6751$ & $0.5769$ &   $0.5772$ & $0.6751$ & $0.6087$ &   $0.6126$ & $0.6694$ & $0.6667$ &   $0.3495$ & $-$\\
			\hline
			\hline
			$R.T(second)$ &  $2.74$ & $2.67$ & $2.41$ &   $3.29$ & $3.30$ & $3.44$ &   $4.69$ & $4.18$ & $4.35$ &   $74.44$ & $-$\\
			\hline
			$M.U(MB)$ & $0.42$ & $0.42$ & $0.42$ &   $0.42$ & $0.42$ & $0.42$ &   $0.42$ & $0.42$ & $0.42$ &   $595.26$ & $-$\\
			\hline
		\end{tabular}
	\label{tab:resultComparison2}
\end{table*}

\begin{figure*}[!ht]
\begin{center}
\includegraphics[width=1.0\textwidth]{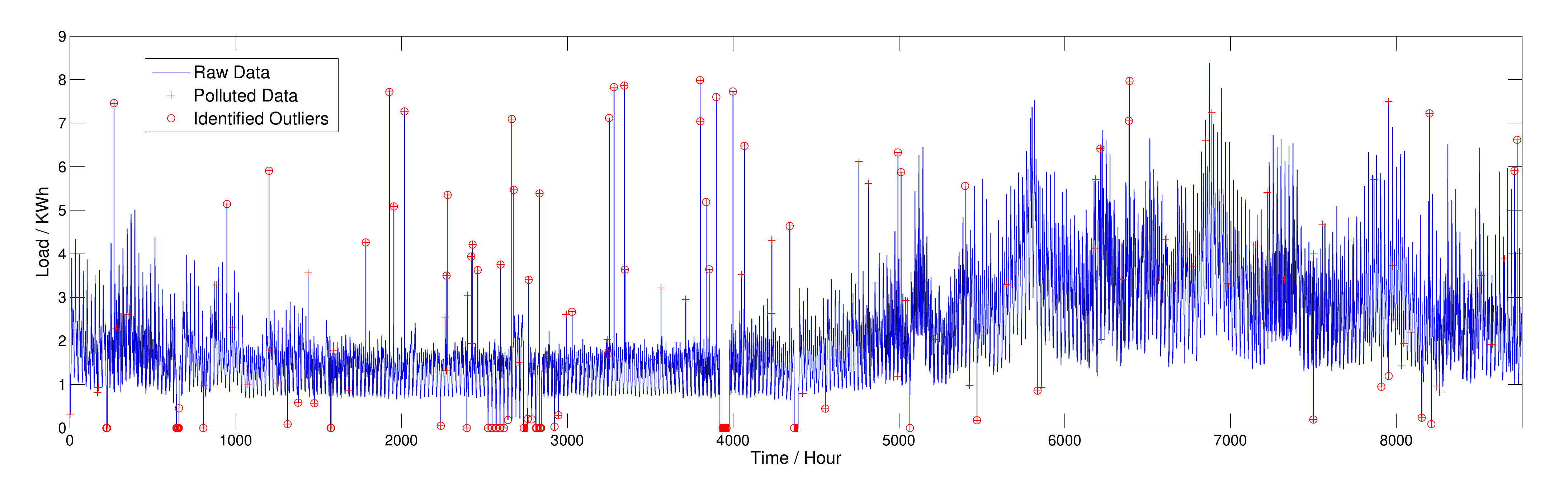}
\caption{Result of outlier detection from gamma distribution based virtual portrait data cleansing ($7$ VLDs)}\label{fig_resultVLD7GB}
\end{center}
\end{figure*}

In our tests, the confidence coefficient is set as $\alpha=0.05$, which results in a confidence interval of $95\%$. Besides, in the IQR-based method, $\rho$ is set as $1.5$, and in the B-spline smoothing method, the degree of freedom ($df$) is treated as a variable and trained when smoothing the load curve.

\subsection{Results from Small-Scale Data}
The one-month data from $01/08/2006$ to $31/08/2006$ in Fig.~\ref{fig_timelinePlot} are firstly used for evaluations. Since these data are stationary, virtual portrait data construction strategy is applied directly. Six virtual portrait datasets are resulted based on the optimal threshold value. \nop{, as shown in Fig.~\ref{fig_6VPD}.

\begin{figure}[!ht]
\begin{center}
\includegraphics[width=3.0in,height=1.9in]{6VPD.pdf}
\caption{Visualization of virtual portrait datasets}\label{fig_6VPD}
\end{center}
\end{figure}
}

In addition to the strategies introduced in Section~\ref{sec:cleansing}, the B-spline smoothing method is also applied for each virtual portrait dataset. Performance metrics are computed with the outcomes from different methods, and the final results are summarized in Table~\ref{tab:resultComparison}. Furthermore, an outcome from IQR-based portrait data cleansing is shown in Fig.~\ref{fig_resultBoxplot}. 

From the above results, we can find that our virtual portrait data cleansing strategies perform much better than B-spline smoothing. For this dataset, both gamma distribution based and IQR-based detection methods perform better than the normal distribution based one. It is interesting to see that applying B-spline smoothing method to virtual portrait data does not bring clear improvement (refer to the results in $5$-th column of Table~\ref{tab:resultComparison}). This implies that using simpler methods on portrait data can achieve good performance already. It is unnecessary to use complex approaches such as B-spline smoothing on portrait data.

Furthermore, we can see that the virtual portrait data cleansing runs faster and uses much less memory than B-spline smoothing. In fact, most time and memory spent in our strategies are on the construction of  virtual portrait datasets, and the overhead of data cleansing over portrait data is negligible. B-spline smoothing, however, spent over $99\%$ of the running time and memory on the calculation of basis functions, which are used to fit the landscape load curve data.

\subsection{Results from Large-Scale Non-Stationary Data}

In practice, the size of load curve data is usually large and covers a time period as long as several years. Therefore, we also test the performance of our method on the one-year data shown in Fig.~\ref{fig_timelinePlot}. 

Note that the landscape data are not always stationary during the whole time window, so we pre-process the data with the method introduced in Section~\ref{sec:bvLandscape}. For comparison, three solutions with 5, 7 and 10 VLDs are provided for tests and evaluations (7 VLDs are resulted from the optimal threshold value). Then for each VLD of each solution, following the same operations for small-scale dataset, we construct its virtual portrait datasets and apply portrait data cleansing strategies to identify outliers.

The results are summarized in Table~\ref{tab:resultComparison2}, and an outcome from 7 VLDs and gamma distribution based portrait data cleansing is shown in Fig.~\ref{fig_resultVLD7GB}.

From the results in Table~\ref{tab:resultComparison2}, we can see that with non-stationary landscape data, virtual portrait data cleansing strategies are still effective and perform well. According to \emph{F-measure}, gamma distribution based cleansing strategy does better than the other two, indicating that it achieves a good balance between precision and recall and has a better overall performance; IQR-based cleansing does better at \emph{precision}, indicating that this strategy performs well at exactness of outlier detection. \nop{Therefore, for large-scale data cleansing, gamma distribution based method should be preferred to normal distribution based one.}

\nop{Normal distribution based data cleansing does better at \emph{precision}, indicating that this strategy performs well at exactness of outlier detection; gamma distribution based cleansing strategy has a better performance at \emph{recall}, showing that it does well at completeness of outlier detection; both IQR-based and gamma distribution based cleansing methods have good results on \emph{F-measure}, indicating that they achieve a good balance between precision and recall. }

In contrast, outlier detection with B-spline smoothing performs poorly. With the largest degree of freedom that the computation allows\footnote{B-spline smoothing with degree of freedom larger than $2815$ is beyond the capability of our desktop computers.}, the \emph{precision}, \emph{recall} and \emph{F-measure} of outlier detection are all below $50\%$. To be worse, the overhead on running time and memory consumption is significantly higher than our method. 

As shown in Fig.~\ref{fig_resultVLD7GB}, our method does not identify many polluted data from time $5000$ till the end. This is because those artificially polluted data are at a comparable level in value as the nearby regular load data. These ``outliers'' are similar to regular values and cannot be effectively detected with any method. 
\nop{
which makes most of them ``hidden outliers'' and 

This, however, does not imply that our method performs badly. The reason why our method fails here is because those artificially polluted data are at a comparable level in value as the surrounded regular load data, which makes most of them ``hidden outliers'' and hard to be captured even with the naked eyes.
}
\subsection{Discussion: Why Does B-spline Smoothing not Perform Well on Load Curve Data?}

\subsubsection{A Local View}

To further investigate why B-spline smoothing does not perform well in load curve data cleansing, we first study the performance in a smaller, local scale. We analyze two situations shown in Fig.~\ref{fig_bsfit148} and Fig.~\ref{fig_bsfit258}, where B-spline smoothing either under-fits or over-fits the load curve data.

From Fig.~\ref{fig_bsfit148}, we can see that the four labeled polluted data were not identified due to the under-fitted regression of B-spline. To alleviate the problem, we may increase the degree of freedom ($df$), but doing so may result in over-fitting. As shown in Fig.~\ref{fig_bsfit258}, in order to fit some outliers (the red dots in the figure), the fitted curve actually deviates from regular data points (the green dots in the figure), which ends up with bad performance in outlier detection. During the process from under-fitting to over-fitting, there must exist a $df$ value which results in the best performance, but finding the best $df$ value is time consuming. 

\begin{figure}[!ht]
\begin{center}
\includegraphics[width=0.5\textwidth,height=0.95in]{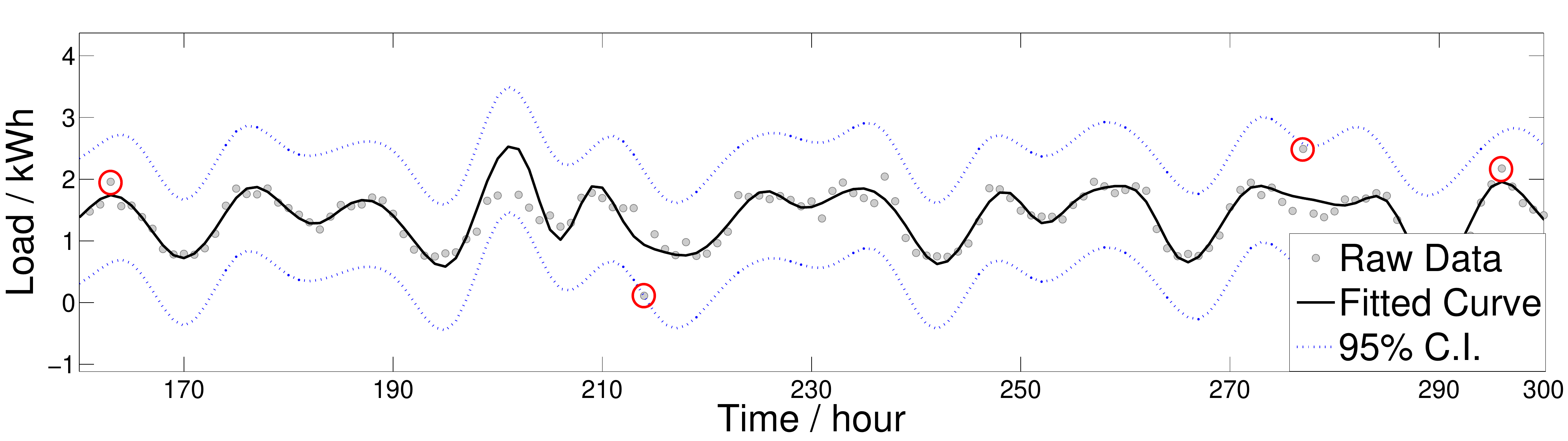}
\caption{Under-fitted B-spline smoothing ($df=100$)}\label{fig_bsfit148}
\end{center}
\vspace{-0.1in}
\end{figure}

\begin{figure}[!ht]
\begin{center}
\includegraphics[width=0.5\textwidth,height=0.95in]{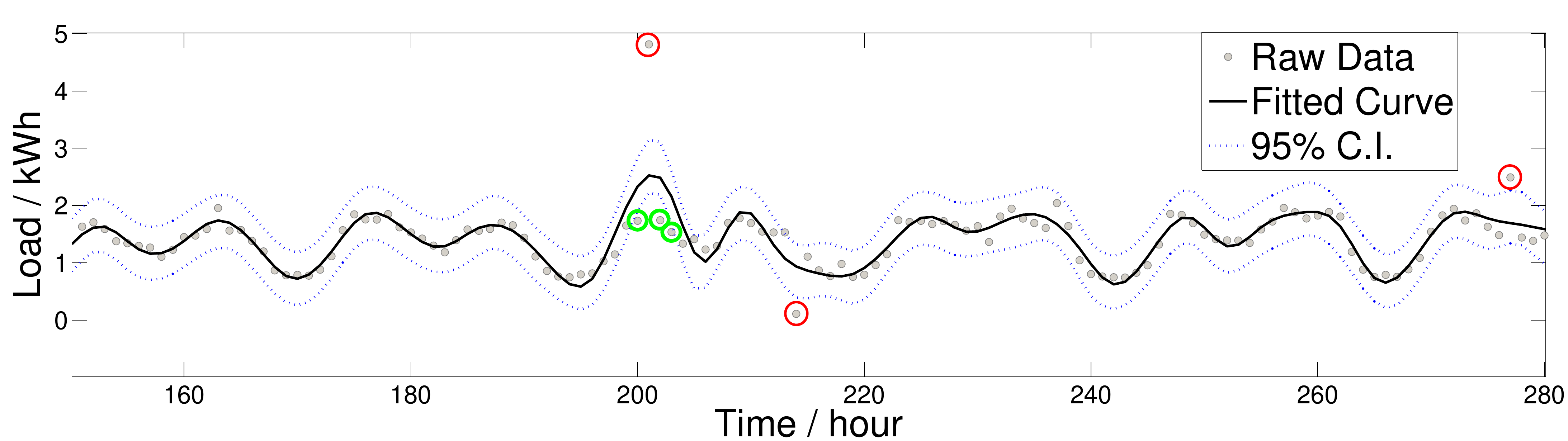}
\caption{Over-fitted B-spline smoothing ($df=200$)}\label{fig_bsfit258}
\end{center}
\vspace{-0.1in}
\end{figure}

The above phenomenon is caused by the inherent problem in regression method, as it treats each data point in the same way and tries to reduce the total estimation error. This may not work well because load curve data at different times follow different statistical features. Using the \emph{portrait} data, in contrast, we can divide data into different groups according to their attributes, and analyze each group separately. Thus ``pathological'' data values may infect landscape data on a large time window but has only limited impact on portrait data. This is the essential point where B-spline smoothing performs poor while virtual portrait data cleansing does better. 

\subsubsection{A Global View}

\begin{figure*}[!ht]
\begin{center}
\includegraphics[width=1.0\textwidth]{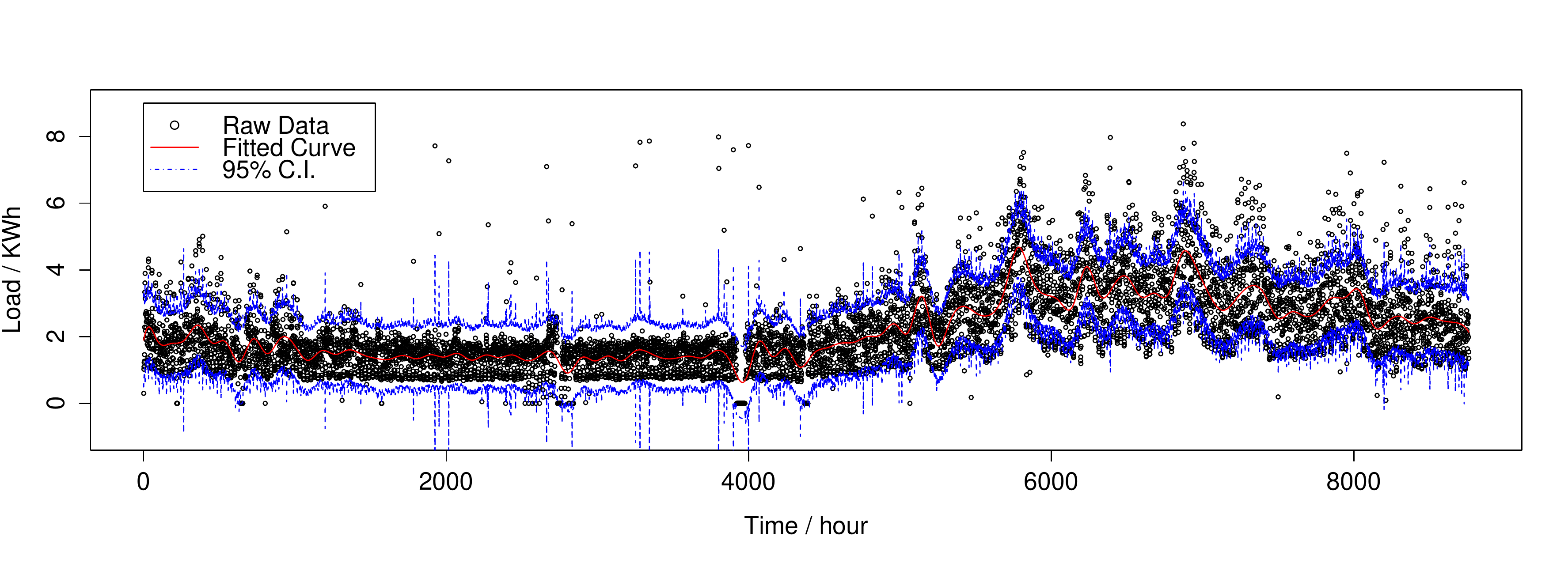}
\caption{Results of B-spline smoothing for large-scale data ($df=100$)}\label{fig_outcomeBS2}
\end{center}
\vspace{-0.1in}
\end{figure*}

An outcome from B-spline smoothing with $df = 100$ is shown in Fig.~\ref{fig_outcomeBS2}, in which we can have a global view of its performance. 

\nop{From $3920h$ to $3975h$ in the load curve, the data are lost and are treated as zeros during B-spline smoothing, as refer to Section~\ref{sec:imputation}. We can find that most missing data are not identified, which is caused by an apparent curvature trend to fit the filled data. We can imagine that, even we replace the consecutively polluted data with other values, such curvature trend is inevitable. This exposes another drawback of regression-based outlier detection methods: they cannot deal with \emph{consecutively} polluted data. In contrast, our portrait data cleansing strategies do not have such a problem. Since all the landscape data will be separated into different portrait data, the consecutively polluted data will be evaluated and handled differently. As a result, one polluted data will not affect nearby ones.}

From $3920h$ to $3975h$ in the load curve, the data are lost and are treated as zeros during B-spline smoothing. We can find that most missing data are not identified. This is caused by an apparent curvature trend to fit the filled data. Even if we replace the missing data with other constants, such a curvature trend is inevitable. This exposes another drawback of regression-based outlier detection methods: they cannot deal with \emph{consecutively} polluted data. In contrast, our portrait data cleansing strategies do not have such a problem. Since all the landscape data will be separated into different portrait data, the consecutively polluted data will be evaluated and handled differently. As a result, one polluted data will not affect nearby ones.

In some special time periods such as holidays, the load values may be consecutively higher or lower in the landscape data. This scenario is similar to the above case. With regression-based outlier detection, there will be an inevitable curvature trend to fit the irregular data, while with our strategies, such data values are separated into different virtual portrait datasets and can be detected with a high possibility. \nop{Furthermore, for a relative small-scale dataset which is more likely stationary time series, irregular data can be identified more easily. }

\section{Conclusion}\label{sec:conclusion}
A new approach was presented to organizing and analyzing load curve data. This approach was based on the inherent periodic patterns in the load curve data and re-organized the data into virtual portrait datasets that could be captured with simple models. Compared to existing regression-based analysis, portrait data based approach significantly simplified many data analysis tasks such as outlier detection. In addition, with simple data pre-processing, our method could effectively handle large-scale non-stationary load curve data. We tested our approach with real-world trace data, including a small-scale stationary dataset and a large-scale non-stationary dataset. The experimental results demonstrated that our approach was much more effective and efficient than existing regression-based methods over both small-scale and large-scale load curve data.


%

\nop{
\appendices
\section{Proof of the First Zonklar Equation}
Appendix one text goes here.

\section{}
Appendix two text goes here.
}

\section*{Acknowledgment}
This research was supported by the Natural Sciences and Engineering Research Council of Canada, National Natural Science Foundation of China (No. 61373152), and Shanghai Committee of Science and Technology, China (No. 13ZR1417500).

\ifCLASSOPTIONcaptionsoff
  \newpage
\fi



%

\bibliographystyle{abbrv}
\bibliography{Reference}

%
\begin{IEEEbiographynophoto}{Guoming Tang}
received his B.S. and M.S. degree from National University of Defense Technology (NUDT), China. He is currently pursuing the Ph.D. degree in the Department of Computer Science, University of Victoria, BC, Canada. His research interests include data cleansing and data mining in smart grid.
\end{IEEEbiographynophoto}

\begin{IEEEbiographynophoto}{Kui Wu}
received the B.Sc. and the M.Sc. degrees in Computer Science from Wuhan University, China in 1990 and 1993,
respectively, and the Ph.D. degree in Computing Science from the University of Alberta, Canada, in 2002. He joined the Department of Computer Science at the University of Victoria, Canada in 2002 and is currently a Professor there. He is also currently an Adjunct Professor in Shanghai University of Electric Power. His research interests include smart grid, mobile and wireless networks, and network performance evaluation.
\end{IEEEbiographynophoto}

\begin{IEEEbiographynophoto}{Jingsheng Lei}
received the Ph.D in Computer Science from Xinjiang University, China. He is currently a Professor and the dean of School of Computer and Information Engineering, Shanghai University of Electric Power. His research interests include data mining, machine learning, and smart grid. 
\end{IEEEbiographynophoto}

\begin{IEEEbiographynophoto}{Zhongqin Bi}
received the Ph.D. degree from East China Normal University in 2009. He is currently an Associate Professor in School of Computer and Information Engineering, Shanghai University of Electric Power. His research interests include smart grid, data quality control, and cloud computing. 
\end{IEEEbiographynophoto}

\begin{IEEEbiographynophoto}{Jiuyang Tang}
received the Ph.D. degree from National University of Defense Technology (NUDT). He is currently an Associate Professor at Institute of Information System and Management of NUDT. His current research interests include big data analysis, P2P, and ubiquitous networks. 
\end{IEEEbiographynophoto}




\end{document}